\documentclass{article}
\usepackage{fullpage}
\usepackage{amsthm}
\usepackage{color}
\usepackage[colorlinks,naturalnames,citecolor=blue,urlcolor=blue]{hyperref}
\newcommand{\myhref}[1]{\href{#1}{\url{#1}}}
\usepackage{authblk}

\usepackage{tikz}
\usetikzlibrary{calc,intersections,through,scopes,arrows,decorations,backgrounds,positioning,fit,petri,automata,shapes}

\newcommand{\lt}[1]{\stackrel{#1}{\longrightarrow}}
\newcommand{\implies}[0]{\Rightarrow}
\newcommand{\cantell}[0]{\mathit{CanTell}}

\newtheorem{theorem}{Theorem}
\newtheorem{lemma}{Lemma}

\newtheorem{definition}{Definition}

\title{Decentralized Observation of Discrete-Event Systems:  At Least One Can Tell\thanks{
This work has been supported by NSF SaTC award CNS-1801546 and by a Discovery Grant from the Natural Sciences and Engineering Research Council of Canada (NSERC).}}
\author[1]{Stavros Tripakis}
\author[2]{Karen Rudie}
\affil[1]{Khoury College of Computer Sciences, Northeastern University}
\affil[2]{Department of Electrical and Computer Engineering, Queen's University}

\begin{document}
\maketitle

\section{Introduction}

The move towards multi-agent, autonomous systems in daily living, including decentralized micro-grids for power, autonomous robots, self-driving cars, and smart small appliances and electronics, requires truly decentralized decision-making, obviating the need for a centralized decision point or fusion rule or coordinator. In this work we explore agents who make decentralized observations and we examine under what conditions the agents' decisions suffice to determine if some behavior is legal or not.

In particular, we introduce a new decentralized observation condition which we call {\em at least one can tell} and which attempts to capture the idea that for any possible behavior that a system can generate, at least one decentralized observation agent can tell whether that behavior was ``good'' or ``bad'', for given formal specifications of ``good'' and ``bad''.
We provide several equivalent formulations of the {\em at least one can tell} condition, and we relate it to (and show that it is different from) previously introduced {\em joint observability}~\cite{TripakisCDC01,TripakisIPL}. In fact, contrary to joint observability which is undecidable~\cite{TripakisCDC01,TripakisIPL}, we show that the {\em at least one can tell} condition is decidable. We also show that when the condition holds, finite-state decentralized observers exist.

\section{Background}

The setting of this work is that of systems whose behavior can be thought of as sequences of actions or events, also referred to as {\em discrete-event systems} (DESs) and modelled as automata or languages. For more details on automata theory and languages, see the classic book by Hopcroft and Ullman \cite{HopcroftUllman79} and for details on discrete-event systems the book by Wonham and Cai \cite{WonhamCai}.

\subsection{Preliminaries}
\label{sec_preliminaries}

A finite set of {\em letters} $\Sigma$ is called an {\em alphabet}.
The set of all finite sequences over $\Sigma$, also called {\em words}, is denoted by $\Sigma^*$.
 The empty word (i.e., the sequence of length $0$) is denoted $\epsilon$.
Given a subalphabet $\Sigma_1 \subseteq \Sigma$, the {\em projection function from $\Sigma$ onto $\Sigma_1$} is the function $P_1 : \Sigma^* \to \Sigma_1^*$ that removes from words in $\Sigma$ all letters except those in $\Sigma_1$.
For example, if $\Sigma=\{a,b\}$ and $\Sigma_1=\{a\}$, then $P_1(abbab) = aa$ and $P_1(bb)=\epsilon$. 
A {\emph language} $L$ over $\Sigma$ is a subset of all possible words and will be used to denote the behavior of some system or process, where elements of $\Sigma$ are events. In that context, the projection function will be used to capture observations of the behavior of that system or process.

A {\em finite automaton over $\Sigma$} is a tuple $A=(\Sigma,Q,Q_0,F,\Delta)$, where $Q$ is a finite set of {\em states};
$Q_0\subseteq Q$ is a set of {\em initial} states; $F\subseteq Q$ is a set of {\em final} or {\em accepting} states; and
$\Delta\subseteq Q\times\Sigma\times Q$ is the {\em transition relation}.
blue{
	When the relation $\Delta$ can be represented as a function $\Delta: Q \times \Sigma \to Q$, we say that $A$ is a {\em deterministic finite automaton} (DFA) and otherwise it is a {\em nondeterministic finite automaton} (NFA).  If we think of elements of $\Sigma$ as events then a DFA is one where at a given state, an event leads to only one other state, whereas an NFA allows for the same event from a state to lead to multiple other states.
}
A {\em run} of $A$ is a finite sequence of states and transitions $q_0 \lt{a_1} q_1 \cdots \lt{a_n} q_n$, 
for $n\ge 0$, such that $q_0\in Q_0$ and 
$(q_i,a_{i+1},q_{i+1}) \in \Delta$ for all $i=0,...,n-1$.
The run is called {\em accepting} if $q_n\in F$.
The run is said to {\em generate} the word $a_1\cdots a_n$.
A word  $\rho\in\Sigma^*$ is said to be {\em accepted} or {\em recognized} by $A$ if there is an accepting run of $A$ generating $\rho$. The language accepted or recognized by $A$ is the set of all words in $\Sigma^*$ accepted by $A$.

	  Since projection can be used to capture the observations of a system, it is natural to ask whether the projection of a language recognized by a finite automaton $A$ can also be recognized by a finite automaton. In fact, the answer is ``yes''. Informally the process is as follows.  Let $A$ be an automaton over $\Sigma$ and let $\Sigma_1\subseteq\Sigma$. We wish to find an automaton $A'$ over $\Sigma_1$ that recognizes $P_1(L(A))$, where $P_1 : \Sigma^* \to \Sigma_1^*$,
	  and for a language $L\subseteq\Sigma^*$, $P_1(L) = \{ P_1(\rho) \mid \rho \in L\}$. We start by replacing all alphabet labels in $A$ by $\varepsilon$.  This results in an NFA $A'$ with $\varepsilon$-transitions. $A'$ already recognizes $P_1(L(A))$ and is a solution to our problem, provided having a non-deterministic automaton is not a problem. If $A'$ is required to be deterministic, we can use the standard subset construction of~\cite{HopcroftUllman79} to convert an NFA with $\varepsilon$-transitions into a DFA. 
	  
	 	 The {\em inverse projection} $P_1^{-1}(L) : \Sigma_1^* \to \Sigma$ of a given language returns all words whose projection is in that language, i.e., $P_1^{-1}(L) = \{ \rho \mid  P_1(\rho) \in L\}$.  Inverse projection allows us to speak about the words an agent thinks could have been produced based on the agent's observations.  Inverse projection of a regular language can also be recognized by a finite automaton.  For $\Sigma_1 \subseteq \Sigma$, let $A$ be an automaton over $\Sigma_1$. To find an automaton that recognizes $P_1^{-1}(L(A))$, add self-loops of all events in $\Sigma \setminus \Sigma_1$ at all states of $A$.

\subsection{Previous work: joint observability (JO)}
\label{sec_jo}

Previous work~\cite{TripakisCDC01,TripakisIPL}  introduced the following
definition of {\em joint observability}:

\begin{definition}[Joint Observability]
Given alphabet $\Sigma$ and subalphabets $\Sigma_i\subseteq\Sigma$,
for $i=1,...,n$, and given regular languages $K\subseteq L\subseteq\Sigma^*$,
$K$ is called {\em jointly observable} if there exists a total function
$$
f : \Sigma_1^*\times\cdots\times\Sigma_n^* \to \{0,1\}
$$
such that
$$
\forall\rho\in L : \rho\in K \iff f(P_1(\rho),...,P_n(\rho)) = 1
$$
where $P_i:\Sigma^*\to\Sigma_i^*$ is the projection function onto $\Sigma_i$, for $i=1,...,n$.
\end{definition}

In the definition above $L$ models the plant, and $K$ models the good behaviors
of the plant. We want to know whether the behavior of the plant was good or
bad. But we can't observe the behavior of the plant directly, so we have to
rely on the decentralized projections. 
Joint observability says that whether a behavior is good can be completely determined based only on the decentralized observations of agents.

The following is a necessary and sufficient condition for joint observability:
\begin{theorem}[\cite{TripakisCDC01,TripakisIPL}]
\label{thm_JO}
$K$ is jointly observable iff
\begin{equation}
\label{eq_JO}
\not\exists\rho,\rho'\in L : \rho\in K\land \rho'\in L-K\land
	\forall i\in\{1,...,n\} : P_i(\rho)=P_i(\rho')
\end{equation}
\end{theorem}
\noindent
Informally, Condition~(\ref{eq_JO}) says that whether strings are good or not cannot possibly be determined if there are two strings---one good and one bad---that look the same to all agents.

We call Condition~(\ref{eq_JO}) the JO condition, or simply JO.
It turns out that JO is undecidable~\cite{TripakisCDC01,TripakisIPL}.

\section{Decentralized Observability: {\em At Least One Can Tell}}

Even though JO has been used in~\cite{TripakisCDC01,TripakisIPL} as a stepping stone to showing the undecidability of decentralized supervisory control problems, JO itself is not really decentralized,
because it requires a centralized decision point $f$. In this paper, we investigate ``truly decentralized'' observation conditions. We begin by a definition that 
tries to capture the {\em at least one can tell} property: namely, that there
is no centralized decision point, but for every behavior of the plant, at
least one of the decentralized observers can tell whether this behavior
is good or bad, i.e., whether it belongs in $K$ or not.
We call this condition {\em OCT}.
We compare OCT to JO, and we examine equivalent formulations of OCT.

\subsection{The OCT condition}

\begin{definition}[At Least One Can Tell]
	\label{def_oct}
Given alphabet $\Sigma$ and subalphabets $\Sigma_i\subseteq\Sigma$,
for $i=1,...,n$, and given regular languages $K\subseteq L\subseteq\Sigma^*$,
	the {\em at least one can tell} condition (or simply OCT)
	is defined as follows:
\begin{equation}
\label{eq_oct}
\forall\rho\in L : \exists i\in\{1,...,n\} : \cantell(i,\rho) 
\end{equation}
where $\cantell$ is defined as follows:
\begin{eqnarray}
\cantell(i,\rho) & = &
	\rho\in K \implies \not\exists\rho'\in L-K : P_i(\rho)=P_i(\rho') \label{eqnoctgood}\\
	& & \land \nonumber \\
	& & \rho\in L-K \implies \not\exists\rho'\in K : P_i(\rho)=P_i(\rho') \label{eqnoctbad}
\end{eqnarray}
\end{definition}

\subsection{The negation of OCT}

In view of what follows, it is useful to state explicitly the negation of the OCT condition:

\begin{lemma}
\label{lem_negationOCT}
The negation of OCT is equivalent to:
\begin{eqnarray}
	& \Big( \exists\rho\in K : \forall i\in\{1,...,n\} : \exists\rho_i\in L-K : P_i(\rho)=P_i(\rho_i) \Big) \label{eqnnegOCTgood}\\
	& \lor \nonumber\\
	& \Big( \exists\rho\in L-K : \forall i\in\{1,...,n\} : \exists\rho_i\in K : P_i(\rho)=P_i(\rho_i) \Big) \label{eqnnegOCTbad}
\end{eqnarray}
\end{lemma}
\begin{proof}
	The negation of OCT is:
$$
\exists\rho\in L : \forall i\in\{1,...,n\} : \neg\cantell(i,\rho)
$$
Since $\rho\in L$ is equivalent to $\rho\in K \lor \rho\in L-K$, the
above is equivalent to:
$$
\Big(\exists\rho\in K : \forall i\in\{1,...,n\} : \neg\cantell(i,\rho) \Big)
\lor
\Big(\exists\rho\in L-K : \forall i\in\{1,...,n\} : \neg\cantell(i,\rho) \Big)
$$
which is by definition of $\cantell$ equivalent to the disjunction of (\ref{eqnnegOCTgood}) and (\ref{eqnnegOCTbad}).
\end{proof}

\subsection{Comparison of OCT with JO}

It is easy to show that OCT is a stronger condition than JO:

\begin{theorem}
	OCT implies JO.
\end{theorem}
\begin{proof}
By contrapositive.
Suppose JO doesn't hold. Then, by Theorem~\ref{thm_JO}, there exist $\rho,\rho'$ such that
$\rho\in K$, $\rho'\in L-K$, and for all $i=1,...,n$, $P_i(\rho)=P_i(\rho')$.
We will show that (\ref{eqnnegOCTgood}) holds. Indeed, this is done by setting $\rho_i$ to $\rho'$ for each $i=1,...,n$.
By Lemma~\ref{lem_negationOCT}, (\ref{eqnnegOCTgood}) implies the negation of OCT.
\end{proof}

The two conditions are {\em not} equivalent, however:

\begin{theorem}
\label{thm_JO_notimply_OCT}
JO does not generally imply OCT.
\end{theorem}
\begin{proof}
Consider the following example:
$\Sigma=\{a,b\}$, $\Sigma_1=\{a\}$, $\Sigma_2=\{b\}$,
$K=(a b)^*$ and $L=(a b)^* b^*$.
We have two observers, and in this case JO holds:
if the numbers of $a$'s and $b$'s are equal, we know that the word was in $K$,
otherwise there are more $b$'s than $a$'s, and the word must have been in
$L-K$. But the decentralized observers alone cannot tell:\footnote{
	We found this out
	thanks to a student in the DES school, Martijn Goorden, who cleverly
observed that in an example presented in Tripakis' lecture, there doesn't
seem to be a way for {\em at least one observer to tell}, yet
joint observability holds. We thank Martijn for this observation.
	}
one observer sees a bunch of $a$'s, the other a bunch of $b$'s. There is
no way of comparing the number of $a$'s and $b$'s (since there is no centralized
decision point). So OCT shouldn't hold here.
Indeed, take $\rho=abb$, $\rho_1=ab$, and $\rho_2=abab$. Note that $\rho\in L-K$, $\rho_1\in K$ and $\rho_2\in K$.
We will show that (\ref{eqnnegOCTbad}) holds. 
Indeed, $P_1(\rho)=P_1(\rho_1)=a$, so observer 1 cannot tell.
Similarly, $P_2(\rho)=P_2(\rho_2)=bb$, so observer 2 cannot tell.
By Lemma~\ref{lem_negationOCT}, (\ref{eqnnegOCTbad}) implies that OCT does not hold.
\end{proof}

The distinction between JO and OCT can also be seen by noticing that a counterexample to JO is a pair $\rho, \rho'$, such that $P_i(\rho)=P_i(\rho')$ for all $i=1,...,n$ (Theorem~\ref{thm_JO}), whereas as Lemma~\ref{lem_negationOCT} indicates and as we saw in the proof of Theorem~\ref{thm_JO_notimply_OCT}, a counterexample to OCT is a set of $n+1$ words, $\rho, \rho_1,...,\rho_n$, such that for each $i=1,...,n$, we have $P_i(\rho)=P_i(\rho_i)$.
Crucially, $\rho'$ is the same word which must ``match'' $\rho$ in every projection, whereas the words $\rho_1,...,\rho_n$ need not be the same.

\subsection{Functional characterization of OCT}
\label{sec_atleastonecantell}

Definition~\ref{def_oct} captures our intuition about the {\em at least one can tell} property, namely, that at least one of the decentralized agents can be sure whether the behavior of the plant was good or bad. However, Definition~\ref{def_oct} does not make explicit the existence of such decentralized observation agents. We rectify this by giving the definition that follows, which we then prove equivalent to OCT.

\begin{definition}[Alternative OCT]
\label{def_altoct}
We say that the {\em alternative OCT} (ALTOCT) condition holds iff there exist total functions $f_i:\Sigma_i^*\to\{Y,N,U\}$, such that
\begin{eqnarray}
	\nonumber
\big( \forall\rho\in L : \exists i\in\{1,...,n\} : & &
	\rho\in K \implies f_i(P_i(\rho))=Y\\
	\nonumber
	& & \land \\
	\nonumber
	& & \rho\in L-K \implies f_i(P_i(\rho))=N \big) \\
\label{eq_atleastonecantell}
	& \land & \\
	\nonumber
\big( \forall\rho\in L : \forall i\in\{1,...,n\} : & & 
				f_i(P_i(\rho))=Y \implies \rho\in K \\
				\nonumber
			& & \land \\
			\nonumber
			& & f_i(P_i(\rho))=N \implies \rho\in L-K \big)
\end{eqnarray}
\end{definition}

$Y$ means that $f_i$ knows that $\rho$ was in $K$, $N$ means that
$f_i$ knows that $\rho$ was not in $K$, and $U$ means $f_i$ doesn't know.
The bottom, $\forall\rho ... \forall i ...$ part says that no observer can
``lie'', namely, if it says $Y$ then it's really the case that $\rho\in K$,
and if it says $N$ then it's really the case that $\rho\in L-K$.
The top, $\forall\rho ... \exists i ...$ part says that at least one observer 
can tell.

\begin{theorem}
\label{thm_aloct}
The ALTOCT condition is equivalent to the OCT condition.
\end{theorem}
\begin{proof}
Suppose the condition of Theorem~\ref{thm_aloct} holds. We will show that
OCT holds. Pick some $\rho\in L$. 
By the first conjunct of Condition~\ref{eq_atleastonecantell}, there is
some $i\in\{1,...,n\}$ such that function $f_i$ satisfies the following:
\begin{eqnarray}
	\label{eq_conjunct1}
	\big(\rho\in K \implies f_i(P_i(\rho))=Y\big)
	& \land &
	\big(\rho\in L-K \implies f_i(P_i(\rho))=N \big)
\end{eqnarray}
We reason by cases:
\begin{itemize}
	\item $\rho\in K$: 
		Then $f_i(P_i(\rho))=Y$. We claim that
		$\cantell(i,\rho)$ holds. Suppose not. Then there exists
		$\rho'\in L-K$ such that $P_i(\rho)=P_i(\rho')$.
		So $f_i(P_i(\rho'))=Y$. But then, by the second conjunct of 
		Condition~\ref{eq_atleastonecantell}, we must have
\begin{eqnarray}
	\label{eq_conjunct2}
      \big(f_i(P_i(\rho'))=Y \implies \rho'\in K \big)
	& \land &
	\big( f_i(P_i(\rho'))=N \implies \rho'\in L-K \big)
\end{eqnarray}
		so $\rho'\in K$, which is a contradiction.
		Therefore $\cantell(i,\rho)$ holds. 
	\item $\rho\in L-K$: 
		Then $f_i(P_i(\rho))=N$. We claim that
		$\cantell(i,\rho)$ holds. Suppose not. Then there exists
		$\rho'\in K$ such that $P_i(\rho)=P_i(\rho')$.
		So $f_i(P_i(\rho'))=N$. But then, by the second conjunct of 
		Condition~\ref{eq_atleastonecantell}, we must again have
		(\ref{eq_conjunct2}),
		so $\rho'\in L-K$, which is a contradiction.
		Therefore $\cantell(i,\rho)$ holds. 
\end{itemize}
In both cases we have established $\cantell(i,\rho)$, which proves OCT.
Thus, the condition of Theorem~\ref{thm_aloct} implies OCT.

We now prove the converse, namely, that OCT implies 
the condition of Theorem~\ref{thm_aloct}. Suppose OCT holds.
In order to establish the condition of Theorem~\ref{thm_aloct} we need
to define total functions $f_i:\Sigma_i^*\to\{Y,N,U\}$ such that
Condition~\ref{eq_atleastonecantell} is satisfied. We define each $f_i$
as follows. Let $\sigma\in\Sigma_i^*$. Then:
\begin{eqnarray}
	\label{eq_functions_aloct}
	f_i(\sigma) = \left\{\begin{array}{l}
		Y \mbox{, if } \exists\rho\in K: P_i(\rho)=\sigma \land
				\forall\rho'\in L-K: P_i(\rho')\ne\sigma \\
		N \mbox{, if } \exists\rho\in L-K: P_i(\rho)=\sigma \land
				\forall\rho'\in K: P_i(\rho')\ne\sigma \\
	U \mbox{, otherwise. } 
	\end{array}
		\right.
\end{eqnarray}
It remains to show that the functions defined in~(\ref{eq_functions_aloct})
satisfy Condition~\ref{eq_atleastonecantell}. We prove each of the two
conjuncts of Condition~\ref{eq_atleastonecantell} separately.

For the first conjunct, pick some $\rho\in L$. 
We must find $i\in\{1,...,n\}$ such that (\ref{eq_conjunct1}) holds.
Pick an $i$ such that $\cantell(i,\rho)$ holds. We know that such an $i$
must exist, by OCT. We now claim that (\ref{eq_conjunct1}) holds:
\begin{itemize}
	\item $\rho\in K$: Then, by $\cantell(i,\rho)$, it must be the
		case that $\not\exists\rho'\in L-K:P_i(\rho)=P_i(\rho')$.
		So $\forall\rho'\in L-K:P_i(\rho)\ne P_i(\rho')$.
		Then, by (\ref{eq_functions_aloct}), $f_i(P_i(\rho))=Y$.
	\item $\rho\in L-K$: Then, by $\cantell(i,\rho)$, it must be the
		case that $\not\exists\rho'\in K:P_i(\rho)=P_i(\rho')$.
		So $\forall\rho'\in K:P_i(\rho)\ne P_i(\rho')$.
		Then, by (\ref{eq_functions_aloct}), $f_i(P_i(\rho))=N$.
\end{itemize}
This proves (\ref{eq_conjunct1}) and the first conjunct of Condition~\ref{eq_atleastonecantell}.

For the second conjunct of Condition~\ref{eq_atleastonecantell},
pick some $\rho\in L$ and some $i\in\{1,...,n\}$.
We must show that (\ref{eq_conjunct2}) holds.
We reason by cases:
\begin{itemize}
	\item $\cantell(i,\rho)$ holds: Then, by the same analysis as above
		we can show that either $\rho\in K$ and $f_i(P_i(\rho))=Y$,
		or $\rho\in L-K$ and $f_i(P_i(\rho))=N$, i.e.,
		that (\ref{eq_conjunct2}) holds.
	\item $\cantell(i,\rho)$ does not hold: Then we claim that
		$f_i(P_i(\rho))=U$. Notice that $\neg\cantell(i,\rho)$ is
		equivalent to
	$$
	\big(\rho\in K \land \exists\rho'\in L-K : P_i(\rho)=P_i(\rho') \big)
        \lor 
	\big(\rho\in L-K \land \exists\rho'\in K : P_i(\rho)=P_i(\rho')\big).
	$$
		Therefore, the first two cases of (\ref{eq_functions_aloct})
		do not hold, so it must be that $f_i(P_i(\rho))=U$.
		Then, (\ref{eq_conjunct2}) holds since both implications
		hold trivially.
\end{itemize}
In both cases, (\ref{eq_conjunct2}) holds. Thus we have shown that OCT
implies the condition of Theorem~\ref{thm_aloct}. This completes the
proof of the theorem.
\end{proof}

Theorem~\ref{thm_aloct} justifies the definition of OCT as a truly 
decentralized observation condition. Indeed, OCT holds if and only if
decentralized functions satisfying condition~(\ref{eq_atleastonecantell})
exist. In fact, we could equivalently have defined OCT to be the existence
of functions satisfying condition~(\ref{eq_atleastonecantell}). Then,
Definition~\ref{def_oct} could be seen as a necessary and sufficient condition
for {\em at least one can tell} observability.

\subsection{DES reformulations of OCT}
\label{sec_alt_reformulation}

This section explores the relationship between OCT and work in the supervisory control of discrete-event systems.  In the DES literature, alphabet symbols represent events and words over an alphabet represent sequences of events. In this section we adopt the convention used in the DES literature of denoting words by letters in the Roman alphabet $s, t, \dots$ and only denoting events by Greek letters, as opposed to the convention in theoretical computing of using Greek letters such as $\rho$ to also denote words.  We have preserved this difference in notation to highlight the point that there is resemblance between our work on OCT which can be situated entirely within theoretical computing (i.e., automata theory and formal languages) without reference to control-theoretic properties and existing work in the supervisory control of DES community.

There are decentralized control conditions within the discrete-event systems literature that bear some resemblance to OCT but are not the same.
Consider the following definition, which on first blush looks similar to a combination of {\em co-observability} \cite{RudieWonhamTAC92} and {\em D\&A co-observability} \cite {YooLafortuneJDEDS02}.    
As we will see in Theorem~\ref{lem_decent_lang_form},
it is actually  decomposability \cite{RudieWonhamPSTV90} (which, informally, says that whether a string is legal can be determined from the observations of legal strings) plus a counterpart to decomposability that roughly says that whether a string is illegal can be determined from the observations of illegal strings.

\begin{definition}[Discrete-Event Systems OCT]
	\label{def_decentobs}
	Given alphabet $\Sigma$ and subalphabets  $\Sigma_1, \Sigma_2,  \dots , \Sigma_n \subseteq\Sigma$, 
	and given regular languages $K\subseteq L\subseteq\Sigma^*$,
	we say that {\em the discrete-event systems OCT (DESOCT) condition} holds if 
	
		\begin{eqnarray}
		\label{eq_decentobs}
		\forall s, s_1, s_2, \ldots , s_n \in \Sigma^* :  
		\big( \bigwedge_{i=1, \ldots , n} P_i(s)=P_i(s_i) \big) & \implies &
		\Big( \big( (\bigwedge_{i=1, \ldots , n}s_i \in K \wedge s\in L) \implies s \in K \big) \nonumber \\
		& & \land  \label{eqn_decentobs}\\ 
		& & \big( (\bigwedge_{i=1, \ldots , n}s_i \in L-K  \wedge s\in L) \implies s \not\in K \big) \Big) \nonumber
	\end{eqnarray}
where $P_i:\Sigma^*\to\Sigma_i^*$ is the projection function onto $\Sigma_i$, for $i=1,...,n$.	
\end{definition}

It can be shown, via logical transformations, that  Definition~\ref{def_decentobs} is equivalent to OCT.

\begin{lemma}
\label{lem_DESOCT_OCT}
	The DESOCT condition is equivalent to the OCT condition.
\end{lemma}	
\begin{proof}

First, (\ref{eq_decentobs}) is equivalent to:
\begin{eqnarray}
		\label{eq_decentobs2}
		\forall s, s_1, s_2, \ldots , s_n \in \Sigma^* :  \big(\bigwedge_{i=1, \ldots , n} P_i(s)=P_i(s_i) \land s\in L\big)
	     & \implies &
	    \big( (\bigwedge_{i=1, \ldots , n}s_i \in K  \implies s \in K) \nonumber \\
		& & \land  \\ 
		& &   (\bigwedge_{i=1, \ldots , n}s_i \in L-K  \implies s \not\in K) \big) \nonumber
\end{eqnarray}
Next, (\ref{eq_decentobs2}) is equivalent to:
\begin{eqnarray}
		\label{eq_decentobs3}
		\forall s \in L, \forall s_1, s_2, \ldots , s_n \in \Sigma^* :  \bigwedge_{i=1, \ldots , n} P_i(s)=P_i(s_i)
		& \implies &
		\big( (\bigwedge_{i=1, \ldots , n}s_i \in K  \implies s \in K) \nonumber \\
		& & \land  \\ 
		& &  (\bigwedge_{i=1, \ldots , n}s_i \in L-K  \implies s \not\in K) \big) \nonumber
\end{eqnarray}
Next, (\ref{eq_decentobs3}) is equivalent to:
\begin{eqnarray}
		\label{eq_decentobs4}
		\forall s \in L, \forall s_1, s_2, \ldots , s_n \in \Sigma^* :  \bigwedge_{i=1, \ldots , n} P_i(s)=P_i(s_i)
		& \implies &
		\big( (\bigwedge_{i=1, \ldots , n}s_i \in K  \implies s \in K) \nonumber \\
		& & \land  \\ 
		& &  (\bigwedge_{i=1, \ldots , n}s_i \in L-K  \implies s \in L- K) \big) \nonumber
\end{eqnarray}
Now let's take the negation of  (\ref{eq_decentobs4}):
\begin{eqnarray}
		\label{eq_decentobs5}
		\exists s \in L, \exists s_1, s_2, \ldots , s_n \in \Sigma^* :  \bigwedge_{i=1, \ldots , n} P_i(s)=P_i(s_i)
		& \land &
		\big( \neg (\bigwedge_{i=1, \ldots , n}s_i \in K  \implies s \in K) \nonumber \\
		& & \lor  \\ 
		& &  \neg (\bigwedge_{i=1, \ldots , n}s_i \in L-K  \implies s \in L- K) \big) \nonumber
\end{eqnarray}
		
Then (\ref{eq_decentobs5}) is equivalent to:
\begin{eqnarray}
		\label{eq_decentobs6}
		\exists s \in L, \exists s_1, s_2, \ldots , s_n \in \Sigma^* :  \bigwedge_{i=1, \ldots , n} P_i(s)=P_i(s_i)
		& \land &
		\big( (\bigwedge_{i=1, \ldots , n}s_i \in K  \land s \not\in K) \nonumber \\
		& & \lor  \\ 
		& &  (\bigwedge_{i=1, \ldots , n}s_i \in L-K  \land s \not\in L- K) \big) \nonumber
\end{eqnarray}
Finally, (\ref{eq_decentobs6}) can be rewritten by observing that $s \in L$ and $s \not\in K$ is
equivalent to $s \in L- K$ (and similarly for $s\in L$ and $s \not\in L-K$):
\begin{eqnarray}
\label{eq_decentobs7}
\exists s \in L, \exists s_1, s_2, \ldots , s_n \in \Sigma^* :  \bigwedge_{i=1, \ldots , n} P_i(s)=P_i(s_i)
& \land &
\big( (\bigwedge_{i=1, \ldots , n}s_i \in K  \land s \in L - K) \nonumber \\
& & \lor  \\ 
& &  (\bigwedge_{i=1, \ldots , n}s_i \in L-K  \land s \in K) \big) \nonumber
\end{eqnarray}
We can see from Lemma~\ref{lem_negationOCT} that (\ref{eq_decentobs7}) is the negation of OCT.   Thus, Definition~\ref{def_decentobs} 
is equivalent to OCT.
\end{proof}

\bigskip
	
Definition~\ref{def_decentobs} can be rewritten as a combination of language inclusions, which will make it immediately apparent that DESOCT is decidable for regular language inputs.

\begin{definition}[Language OCT]
	\label{def_langOCT}
	Given alphabet $\Sigma$ and subalphabets  $\Sigma_1, \Sigma_2,  \dots , \Sigma_n \subseteq\Sigma$, 
	and given regular languages $K\subseteq L\subseteq\Sigma^*$,
	we say that {\em the language OCT (LANGOCT) condition} holds if
	\begin{eqnarray}
	& \bigcap\limits_{i=1, \dots , n} P_i^{-1}(P_i(K)) \cap L  \subseteq K \label{eq_subsetK}\\
	& \land  \nonumber \\
	& \bigcap\limits_{i=1, \dots , n} P_i^{-1}(P_i(L-K)) \cap L  \subseteq L-K \label{eq_subsetNotK}
	\end{eqnarray}
	
\end{definition}

\begin{theorem}  
	\label{lem_decent_lang_form}
	The DESOCT condition is equivalent to the LANGOCT condition.
 \end{theorem}

\begin{proof}  
We will show that (\ref{eq_decentobs}) implies the conjunction (\ref{eq_subsetK}) $\land$ (\ref{eq_subsetNotK}) and vice versa.

\medskip

\noindent
\underline{(\ref{eq_decentobs}) implies (\ref{eq_subsetK}) $\land$ (\ref{eq_subsetNotK})}\\

	Consider $s\in \bigcap\limits_{i=1, \dots , n} P_i^{-1}(P_i(K)) \cap L $. Then $s\in L$ and $s\in P_i^{-1}(P_i(K))$, i.e., $P_i(s)\in P_i(K)$, for all $i=1,...,n$.
	This means that there exist $s_1, s_2, \ldots s_n\in K$ such that $P_i(s_i)=P_i(s)$ for all $i=1,...,n$.  By the first conjunct in the consequent of (\ref{eq_decentobs}), this means that $s \in K$.  Therefore, (\ref{eq_subsetK}) holds.

	Similarly, consider $s\in \bigcap\limits_{i=1, \dots , n} P_i^{-1}(P_i(L-K)) \cap L $. Then there exist $s_1, s_2, \ldots s_n$ such that $P_i(s_i)=P_i(s)$ and $s_1, s_2, \dots, s_n \in L-K$ and $s \in L$.  By the second conjunct in the consequent of (\ref{eq_decentobs}), this means that $s \in L-K$.  Therefore, (\ref{eq_subsetNotK}) holds.

\medskip

\noindent	
\underline{(\ref{eq_subsetK}) $\land$ (\ref{eq_subsetNotK}) implies (\ref{eq_decentobs})}\\

	Now consider $s, s_1, \ldots, s_n \in \Sigma^*$ that satisfy the antecedent of (\ref{eq_decentobs}), i.e., $P_i(s)=P_i(s_i)$ for $i= 1, \ldots , n$.   First, suppose that for all $i= 1, \ldots , n$, $s_i \in K$ and $s \in L$.  Since, for all $i=1, \ldots , n$,  $P_i(s)=P_i(s_i)$ and $s_i \in K$, this means that $s \in P_i^{-1}(P_i(K))$.   Since $s \in \bigcap\limits_{i=1, \dots , n} P_i^{-1}(P_i(K))$ and $s \in L$, by (\ref{eq_subsetK}), we have $s \in K$.

	Second, suppose that for all $i= 1, \ldots , n$, $s_i \in L-K$ and $s \in L$.  Since, for all $i=1, \ldots , n$,  $P_i(s)=P_i(s_i)$ and $s_i \in L-K$, this means that $s \in P_i^{-1}(P_i(L-K))$.   Since $s \in \bigcap\limits_{i=1, \dots , n}P_i^{-1}(P_i(L-K))$ and $s \in L$, by (\ref{eq_subsetNotK}), we have $s \in L-K$.
\end{proof}

\section{Decidability, computational complexity, finite implementation}

We now show that OCT is decidable for regular languages and we provide an asymptotic computational complexity analysis. 
We also show that, when the OCT condition is met, there exists a finite implementation.

\subsection{OCT decidability and computational complexity}
\label{sec_OCT_decidable}

	Using the formulation of OCT given by LANGOCT ((\ref{eq_subsetK}) and (\ref{eq_subsetNotK})), 
	we can demonstrate decidability of OCT. In addition, we will show that OCT can be decided in  time $O(p^{n+2} \cdot m^{n+1})$ where $n$ is the number of agents, and $p$ and $m$ are the number of states of the automata recognizing the languages $L$ and $K$, respectively.

	\begin{theorem}
		\label{thm_OCT-poly-time}
		If $K$ and $L$ are regular languages, OCT is decidable.  Moreover, if $K$ and $L$ are recognized by DFAs
	    whose state sets have cardinality $m$ and $p$, respectively, then for $n$ agents, OCT is decidable in 
		time $O(p^{n+2} \cdot m^{n+1})$.
	\end{theorem}

\begin{proof}
Decidability follows from Lemma~\ref{lem_DESOCT_OCT}, Theorem~\ref{lem_decent_lang_form}, 
	and the fact that the operations of projection, inverse projection, intersection, and checking set containment are decidable for regular languages.

	Let us now examine the computational complexity. 
	Let $K$ be recognized by a DFA $A_K$ with $m$ states and let $L$ be recognized by a DFA $A_L$ with $p$ states. 
	We can compute a DFA $A_{L-K}$ recognizing $L-K$ by noting that $L-K = L\cap \overline{K}$,  where 
	$\overline{K}$ denotes the complement of set $K$.\footnote{
	In the DES literature the overbar notation is used to denote prefix-closure. Here, we are adopting the standard math convention of using overbar to denote set complement.
	}
	$\overline{K}$ is recognized by a DFA $A_{\overline{K}}$ with exactly the same states as $A_K$, as it suffices to turn the accepting states of $A_K$ into rejecting states and vice-versa. Then, $A_{L-K}$ can be built as the Cartesian product of $A_L$ and $A_{\overline{K}}$~\cite{HopcroftUllman79}. The number of states of $A_{L-K}$ is $p\cdot m$.
	
	Consider first condition (\ref{eq_subsetK}).
	The language $P_i(K)$ 
	is recognized by the same automaton that recognizes $K$ but with all events not in $\Sigma_i$ replaced by $\varepsilon$. This process can be completed in $O(m)$ time and the number of states of the resulting automaton is still $m$. The inverse projection $P_i^{-1}(P_i(K))$ is achieved by adding in self-loops of events not in $\Sigma_i$ to the resulting automaton. This can be done in $O(m)$ time and the number of states of the resulting automaton is still $m$. The intersection of the $n$ terms $P_i^{-1}(P_i(K))$, for $i=1, \dots n$, together with $L$ can be done in $O(m^n \cdot p)$ time, since intersection of automata can be represented with an automaton whose state space is the Cartesian product of the constituent automata. The result is an NFA $A_1$ with $m^n \cdot p$ states, such that
	$L(A_1) = \bigcap\limits_{i=1, \dots , n} P_i^{-1}(P_i(K))\cap L$. For condition (\ref{eq_subsetK}) we need to check whether $L(A_1) \subseteq K$. This is equivalent to checking $L(A_1) \cap \overline{K} = \emptyset$, which in turn amounts to computing the product of $A_1$ with $A_{\overline{K}}$.\footnote{
	It is well-known that checking subset inclusion for languages $L_1$ and $L_2$ can done in polynomial time when $L_1$ is represented by an NFA and $L_2$ by a DFA. A table in~\cite{Clemente2020OnTC} provides a nice summary of the computational complexity of checking subset inclusion for various different automata representations of $L_1$ and $L_2$.
	}
	Therefore, the overall time
	complexity of checking condition (\ref{eq_subsetK}) is $O(m^n \cdot p\cdot m) = O(m^{n+1} \cdot p)$.
	
	Next, consider condition (\ref{eq_subsetNotK}). Reasoning similarly as above, we first construct an NFA $A_2$ such that
	$L(A_2) = \bigcap\limits_{i=1, \dots , n} P_i^{-1}(P_i(L-K))\cap L$. This can be done in $O((p\cdot m)^n \cdot p)$ time and the resulting automaton $A_2$ has $(p\cdot m)^n \cdot p$ states. It remains to check whether $L(A_2)\subseteq L-K$, which is equivalent to checking $L(A_2)\cap \overline{L-K} = \emptyset$. Recall that $L-K$ is recognized by DFA $A_{L-K}$ with $p\cdot m$ states. Because $A_{L-K}$ is deterministic, the complement $\overline{L-K}$ is also recognized by a DFA $A_{\overline{L-K}}$ with exactly the same states. Therefore, checking $L(A_2)\subseteq L-K$ can be done by building the product of $A_2$ with $A_{\overline{L-K}}$, where this product has $(p\cdot m)^n \cdot p \cdot (p\cdot m)$ states. Thus, the overall complexity for checking condition (\ref{eq_subsetNotK}) is $O\big((p\cdot m)^n \cdot p \cdot (p\cdot m) \big)= O(p^{n+2} \cdot m^{n+1})$.
	
	We can see that the complexity of condition (\ref{eq_subsetNotK}) is higher than that of condition (\ref{eq_subsetK}), therefore, the overall complexity of OCT is $O(p^{n+2} \cdot m^{n+1})$.
\end{proof}

\subsection{Finite-state OCT observers}

The functions $f_i$ in Theorem~\ref{thm_aloct} can be seen as decentralized
observers, each outputting $Y,N,U$ depending on whether they can tell whether
the original behavior $\rho$ was in $K$, not in $K$, or unknown.
Each of these functions takes as input the entire projection $P_i(\rho)$ which
is generally unbounded in length. So a brute-force implementation of these
functions requires unbounded memory. In this section, we show that we can
also implement these functions using finite memory.

A {\em finite-state observer} is a DFA $O_i$ over
subalphabet $\Sigma_i$, such that the states of $O_i$ are labeled by one
of $Y,N,U$, corresponding to the three observation outcomes of function $f_i$.
The requirement is that for every $\sigma\in\Sigma_i^*$, the label of the
state that $O_i$ ends up in after reading $\sigma$ is exactly $f_i(\sigma)$.
Note that $O_i$ is deterministic, so for a given $\sigma$ there is a unique
state that $O_i$ ends up in after reading $\sigma$.

\begin{theorem}
	\label{thm_finite-state-observer}
If OCT holds, then finite-state observers as above exist.
\end{theorem}
\begin{proof}
We build $O_i$ as follows:
\begin{itemize}
	\item Let $P_i(K)$ be the projection of regular language $K$ onto
		$\Sigma_i$. $P_i(K)$ is a regular language. Let $A_1$ be
		a DFA recognizing $P_i(K)$. Without loss of generality we
		can assume that $A_1$ is complete, meaning there is a
		transition from every state of $A_1$ for every input letter.
	\item Similarly, let $A_2$ be a DFA recognizing $P_i(L-K)$.
	\item Both $A_1,A_2$ are DFA over the same alphabet $\Sigma_i$.
		Let $A$ be the synchronous product of $A_1$ and $A_2$.
		Synchronous product means that each transition of $A$
		corresponds to a pair of transitions for each of $A_1,A_2$,
		labeled with the same letter.
	\item A state $q$ of $A$ is a pair of states $(q_1,q_2)$, where $q_1$
		is a state of $A_1$ and $q_2$ is a state of $A_2$. We label
		$q$ as follows:
		\begin{itemize}
			\item $q$ is labeled with $Y$ if $q_1$ is accepting
				and $q_2$ is rejecting.
			\item $q$ is labeled with $N$ if $q_1$ is rejecting
				and $q_2$ is accepting.
			\item $q$ is labeled with $U$ otherwise.
		\end{itemize}
\end{itemize}
We claim that $O_i$ as constructed above satisfies our requirements
for a correct local observer.
To see this, suppose that OCT holds and consider some 
$\sigma \in \Sigma_i^*$.
We distinguish cases:

\begin{itemize}
\item
There is some $\rho\in L$ such that $\sigma=P_i(\rho)$.
We distinguish subcases:
\begin{itemize}
	\item $\rho\in K$: Suppose that $O_i$, after reading the input word $\sigma$, ends up in a state $q=(q_1,q_2)$.
	Because $\rho\in K$, we have $\sigma\in P_i(K)$, so $q_1$ is accepting.
		There are two subcases:
		\begin{itemize}
			\item $q_2$ is accepting: This means that $\sigma$ is
				also in $P_i(L-K)$. So there must exist some
				$\rho'\in L-K$ such that $P_i(\rho')=\sigma$,
				so $P_i(\rho')=P_i(\rho)$. This means that
				$\cantell(i,\rho)$ does not hold and that
				$f_i(\sigma)$ must be equal to $U$.
				Indeed, $q$ is also labeled $U$.
			\item $q_2$ is rejecting: This means that 
				$\sigma\not\in P_i(L-K)$, which implies that
				$\forall\rho'\in L-K:P_i(\rho')\ne\sigma$,
				$\forall\rho'\in L-K:P_i(\rho')\ne P_i(\rho)$.
				This means that
				$\cantell(i,\rho)$ holds and that
				$f_i(\sigma)$ must be equal to $Y$.
				Indeed, $q$ is also labeled $Y$.
		\end{itemize}
	\item  $\rho\in L-K$: the analysis is similar to the case $\rho\in K$ and is omitted.
\end{itemize}
\item There is no $\rho\in L$ such that $\sigma=P_i(\rho)$.
	Then, observe that Condition~(\ref{eq_atleastonecantell}) does not impose any restrictions on what $f_i(\sigma)$ can be, meaning that function $f_i$ can return any of the three symbols $Y,N,U$.
	As a result, any behavior of $O_i$ on $\sigma$ is  acceptable.
\end{itemize}
\end{proof}

	We see in the construction in the proof of Theorem~\ref{thm_finite-state-observer} that we produce DFA recognizing $P_i(K)$ and $P_i(L-K)$. While we have not proven that a polynomial-time algorithm for finite-state observers does not exist, we speculate that it does not since the first conjunct in the consequent of (\ref{eq_decentobs}) is the same expression that appears in co-observability and producing finite-state supervisors for co-observable systems cannot be done in polynomial time~\cite{RudieWillemsTAC95}.
	The result from Theorem~\ref{thm_OCT-poly-time} and the proof from Theorem~\ref{thm_finite-state-observer}
	are in keeping with results in the field of DES where, if the number of agents is fixed, verification of the necessary and sufficient conditions for supervisory control solutions to exist can be decided in polynomial time in the size of the state sets but where, even when the conditions are satisfied, the synthesis of corresponding supervisors cannot be done in polynomial time (cf., \cite{TsitsiklisMCSS1989} for centralized control and \cite{RudieWillemsTAC95} for decentralized control).

\section{Related Work}

Decentralized observation problems similar to the ones we examine here are also studied in~\cite{TripakisCDC05}.
In particular, the so-called {\em local} observation problems defined in~\cite{TripakisCDC05} are similar in spirit to the ALTOCT condition, but with a crucial difference. In ALTOCT, the local decision functions $f_i$ can each return three possible values, $Y,N,U$, whereas in the local observation problems defined in~\cite{TripakisCDC05}, the local decision functions $f_i$ are only allowed two possible return values, $0$ or $1$. Another difference is that the local observation problems defined in~\cite{TripakisCDC05} include a global combination function $B$ which can be any Boolean function, whereas in our setting of OCT and ALTOCT, the combination is implicitly disjunction: if at least one local observer says $Y$ then this is enough to ensure that the behavior $\rho$ of the plant was good, and if at least one local observer says $N$ then we can be sure that $\rho$ was bad. (By definition, it is impossible to have the case where some observer says $Y$ and another says $N$.)

Decentralized control problems under partial observation are investigated in \cite{CieslakEtAlTAC88,RudieWonhamTAC92,YooLafortuneJDEDS02}. For decentralized discrete-event systems problems that require control, problem solutions require that the agents' observations together with their control capabilities are enough to effect the necessary control. Co-observability and other variations (such as D \& A co-observability~\cite{YooLafortuneJDEDS02}) differ from joint observability and the one can tell condition in the following way.  Co-observability roughly says ``based on what sequence of events has occurred so far, can decentralized supervisors know enough about an upcoming event to know whether to prevent it from occurring''. Together with controllability, co-observability ensures that decentralized supervisory control problems can be solved because for any event that could lead somewhere illegal, at least one agent that \emph{can} stop that event from occuring knows enough from the agent's observations to do so.  As such, co-observability and conditions like it typically are expressed in the form ``for all $s$ that is in $K$ \ldots if $P(s)=$ \ldots and some conditions on string $s$ and on event $\sigma$, then $ s \sigma \in K$''.   The conditions on $s$ and $\sigma$ capture what is necessary and sufficient to determine if $s \sigma$ is in $K$ (or in $L-K$ for D \& A co-observability). These conditions are used so that supervisors make decisions about each upcoming event $\sigma$ and enact control over $\sigma$.  In contrast, joint observability, one can tell and decentralized observability, are divorced from a particular control problem at hand. Rather, they speak to whether a string that has already occurred is or isn't in $K$ and whether or not decentralized agents are able to determine that.

The construction in \cite{RudieWillemsTAC95} used to show that co-observability can be verified in polynomial time also uses a product construction that bears some resemblance to the one used to show that OCT is decidable. In both constructions, the paths through the states in the Cartesian product produce the strings that serve as a violation of the relevant condition (OCT here and co-observability in \cite{RudieWillemsTAC95}). In \cite{RudieWillemsTAC95}, to check co-observability you must check strings $s \sigma$, $s' \sigma$ and $s'' \sigma$ and determine if $s$ is in $L$ before you can even check if $s \sigma \in K$.  So the product automaton in \cite{RudieWillemsTAC95} requires an extra (n+2)th element in its Cartesian product that captures the states of $L$, which is not needed here.  In \cite{RudieWillemsTAC95} it is not assumed that the FSA representing languages are complete so instead at each move it is necessary to check if there is a transition on $\sigma$. An additional dump state $d$ is added to the product there; the purpose of $d$ is to keep track of whether or not co-observability is violated by the strings thus far generated; in contrast, in our construction here we take the $n$ automata $A^1_{L-K} \dots A^n_{L-K}$ to be recognizers of $L-K$ and not $K$ and so a violation of OCT is determined by the presence of strings that land at a terminal state of the product.

\section{Conclusions}

We have shown that checking decentralized observability amounts to ensuring that at least one agent can tell if a word is legal or not. This condition is decidable for regular languages and if it is satisifed then finite-state observers can be constructed (albeit not likely in polynomial time).


\end{document}